\documentclass[11pt,a4paper]{article}
\usepackage{authblk}

\usepackage{latexsym, amsmath, amsthm, amssymb}
\usepackage[usenames,dvipsnames]{color}
\usepackage{graphicx}
\usepackage{mleftright}
\usepackage{cite}

\theoremstyle{plain}
\newtheorem{theorem}{Theorem}
\newtheorem{lemma}[theorem]{Lemma}
\newtheorem{corollary}[theorem]{Corollary}

\theoremstyle{definition}
\newtheorem{definition}{Definition}
\newtheorem{example}{Example}

\newenvironment{myindentpar}[1]%
  {
    \begin{list}{}%
         {
	  \setlength{\leftmargin}{#1}}%
          \item[]%
	 }
    {\end{list}
  }

\title{Conditional entropy of ordinal patterns}
\author[1,2]{Anton M.~Unakafov\thanks{Corresponding address: Institute of Mathematics, University of L\"ubeck,
Ratzeburger Alley 160, Building 64, 23562 L\"ubeck, Germany. Tel.: +49 451 500 4183; fax: +49 451 500 3373. e-mail: anton@math.uni-luebeck.de (A.M. Unakafov)}}
\author[1]{Karsten Keller}
\affil[1]{Institute of Mathematics, University of L\"ubeck}
\affil[2]{Graduate School for Computing in Medicine and Life Sciences, University of L\"ubeck}
\date{November 11, 2013}
\begin{document}
\maketitle

\begin{abstract}
\noindent In this paper we investigate a quantity called conditional entropy of ordinal patterns, akin to the permutation entropy.
The conditional entropy of ordinal patterns describes the average diversity of the ordinal patterns succeeding a given ordinal pattern.
We observe that this quantity provides a good estimation of the Kolmogorov-Sinai entropy in many cases.
In particular, the conditional entropy of ordinal patterns of a finite order coincides with the Kolmogorov-Sinai entropy for periodic dynamics and for Markov shifts over a binary alphabet.
Finally, the conditional entropy of ordinal patterns is computationally simple and thus can be well applied to real-world data.
\\

\end{abstract}

\section{Introduction}\label{intro}
The question how can one quantify the complexity of a system often arises in various fields of research.
On the one hand, theoretical measures of complexity like the Kolmogorov-Sinai (KS) entropy \cite{Walters1982, Choe2005}, the Lyapunov exponent \cite{Choe2005} and others 
are not easy to estimate from given data.
On the other hand, empirical measures of complexity often lack of a theoretical foundation, 
see for instance the discussion of the renormalized entropy and its relationship to the Kullback-Leibler entropy in
\cite{SaparinWittKurthsAnishchenko1994, QuianQuirogaArnholdLehnertzGrassberger2000, KopitzkiWarnkeSaparinKurthsTimmer2002, QuianQuirogaArnholdLehnertzGrassberger2002}. 
Sometimes they are also not well interpretable, for example, see \cite{RichmanMoorman2000} for a criticism of the approximate entropy interpretability.

One of possible approaches to measuring complexity is based on ordinal pattern analysis \cite{BandtPompe2002, KellerSinnEmonds2007, Amigo2010}.
In particular, the permutation entropy of some order $d$ can easily be estimated from the data 
and has a theoretical counterpart (for order $d$ tending to infinity), which is a justified measure of complexity. 
However, in this paper we consider another ordinal-based quantity, the conditional entropy of ordinal patterns. 
We show that for a finite order $d$ in many cases it is closer to the KS entropy than the permutation entropy.

The idea behind ordinal pattern analysis is to consider order relations between values of time series instead of the values themselves.  
The original time series is converted to a sequence of ordinal patterns of an order $d$, 
each of them describing order relations between $(d+1)$ successive points of the time series, as demonstrated in Figure~\ref{figureOrdinalPattens} for order $d = 3$.

\begin{figure}[h]
      \centering
      \includegraphics[scale=0.137]{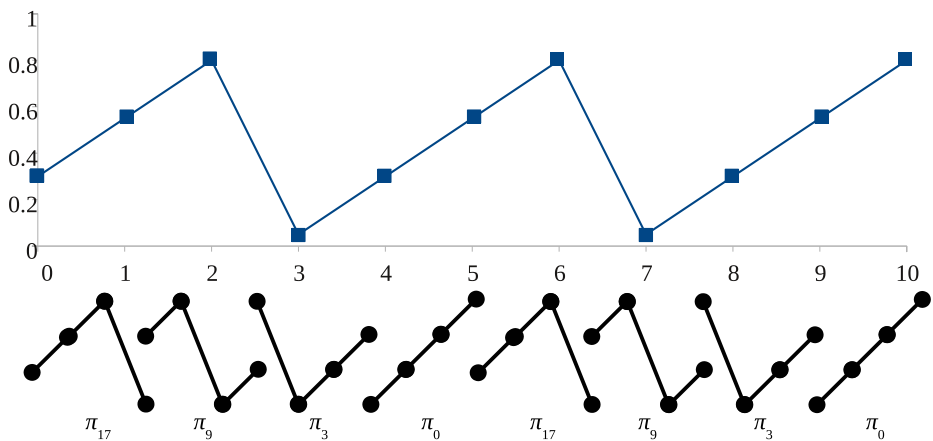}
      \caption{Ordinal patterns of order $d = 3$ for a periodic time series, four patterns occur with period $4$.}
      \label{figureOrdinalPattens}
\end{figure}
 
The more complex the underlying dynamical system is, the more diverse the ordinal patterns occurring for the time-series are.
This diversity is just what the permutation entropy measures.
For example, in Figure~\ref{figureOrdinalPattens} the permutation entropy of order $d=3$ is equal to $\frac{1}{3}\ln4$, 
since there are four different ordinal patterns occurring with the same frequency.
The permutation entropy is robust to noise \cite{Amigo2010}, computationally simple and fast \cite{KellerSinnEmonds2007}.
For order $d$ tending to infinity the permutation entropy is connected to the central theoretical measure of complexity for dynamical systems:
it is equal to the KS entropy in the important particular case \cite{BandtKellerPompe2002}, and it is not lower than the KS entropy in a more general case \cite{Keller2012}.

However, the permutation entropy of finite order $d$ does not estimate the KS entropy well, while being an interesting practical measure of complexity.
Even if the permutation entropy converges to the KS entropy as order $d$ tends to infinity, the permutation entropy of finite $d$ can be either much higher or much lower than the KS entropy 
(see Subsection~\ref{sectOrdinalKS} for details).

Therefore we propose to consider the conditional entropy of ordinal patterns of order $d$:
as we demonstrate, in many cases it provides a much better practical estimation of the KS entropy than the permutation entropy, while having the same computational efficiency.
The conditional entropy of ordinal patterns characterizes the average diversity of ordinal patterns succeeding a given one.
For the example in Figure~\ref{figureOrdinalPattens} the conditional entropy of ordinal patterns of order $d=3$ is equal to zero 
since for each ordinal pattern only one successive ordinal pattern occurs  
($\pi_9$ is the only successive ordinal pattern for $\pi_{17}$, $\pi_3$ is the only successive ordinal pattern for $\pi_{9}$ and so on).

Let us motivate the discussion of the conditional entropy of ordinal patterns by an example.
\begin{example}\label{LogisticFamilyEx}
  Consider the family of logistic maps $f_r:[0,1]\hookleftarrow$ defined by $f_r(x)=rx(1-x)$. 
  For almost all $r \in [0, 4]$ the KS entropy either coincides with the Lyapunov exponent if it is positive or is equal to zero otherwise
  (this holds by Pesin's formula \cite[Theorems~4,~6]{Young2003}, due to the properties of $f_r$-invariant measures \cite{MartensNowicki2000}).
  Note that the Lyapunov exponent for the logistic map can be estimated rather accurately \cite{Sprott2003}.
  For the logistic maps the permutation entropy of order $d$ converges to the KS entropy as $d$ tends to infinity.
  However, Figure~\ref{figureLogisticFamily} shows that for $r \in [3.5, 4]$ the permutation entropy of order $d = 9$ 
  is relatively far from the Lyapunov exponent in comparison with the conditional entropy of ordinal patterns of the same order
  (values of both entropies are numerically estimated from orbits of length $L = 4\cdot 10^6$ of a `random point' in $[0,1]$).
  \begin{figure}[h]
      \centering
      \includegraphics[scale=0.48]{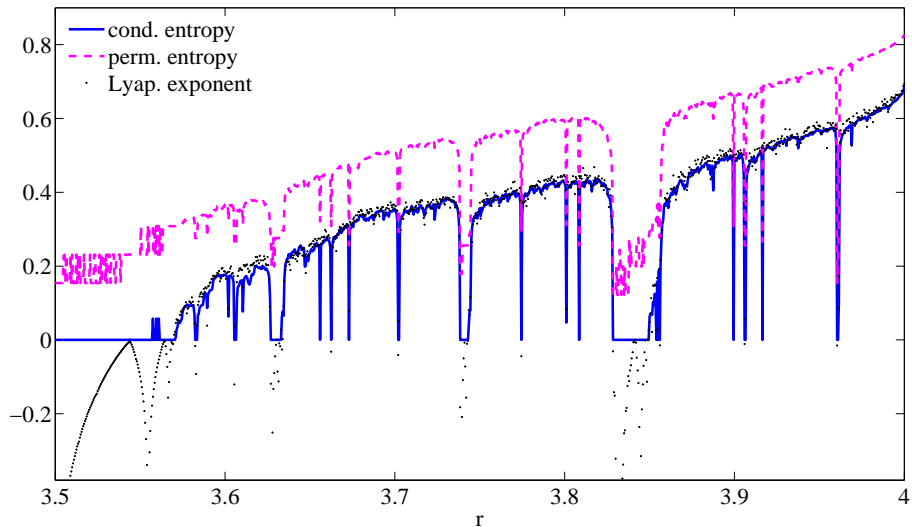}
      \caption{Empirical conditional entropy and permutation entropy in comparison with the Lyapunov exponent for logistic maps}
      \label{figureLogisticFamily}
  \end{figure}
\end {example}

In this paper we demonstrate that under certain assumptions the conditional entropy of ordinal patterns estimates the KS entropy better than the permutation entropy (Theorem~\ref{CEofOP_andOrdinalEntr}).
Besides, we prove that for some dynamical systems the conditional entropy of ordinal patterns for a finite order $d$ coincides with the KS entropy (Theorems~\ref{MarkovShift_cor},~\ref{AllPeriodicPointsTh}),
while the permutation entropy only asymptotically approaches the KS entropy.

The paper is organized as follows.
In Section~\ref{sectionPrelim} we fix the notation, recall the definition of the KS entropy and basic notions from ordinal pattern analysis.
In Section~\ref{sectionCEofOPintro} we introduce the conditional entropy of ordinal patterns and show that in some cases it approaches the KS entropy faster than the permutation entropy.
Moreover, we prove that the conditional entropy of ordinal patterns for finite order $d$ coincides with the KS entropy 
for Markov shifts over two symbols (Subsection~\ref{sectMarkov}) and for systems with periodic dynamics (Subsection~\ref{sectPeriodicCase}).
In Section~\ref{sectionCEofOP_PE_SE} we consider the interrelation between the conditional entropy of ordinal patterns, the permutation entropy and the sorting entropy \cite{BandtPompe2002}.
In Section~\ref{sectionConclusion} we observe some open question and make a conclusion.
Finally, in Section~\ref{sectionProof} we provide those proofs that are mainly technical.

\section{Preliminaries}\label{sectionPrelim}
\subsection{Kolmogorov-Sinai entropy}\label{sectSDandKS}
In this subsection we recall the definition of the KS entropy of a dynamical system and define some related notions we will use further on.
Throughout the paper we use the same notation as in \cite{KellerUnakafovUnakafova2012} and refer to this paper for a brief introduction.
For a general discussion and details we refer the reader to \cite{Choe2005, Kitchens1998}.

We focus on a measure-preserving dynamical system $\left( \Omega, \mathbb{B}(\Omega), \mu, T \right)$, where
$\Omega$ is a non-empty topological space,
$\mathbb{B}(\Omega)$ is the Borel sigma-algebra on it,
$\mu:\mathbb{B}(\Omega) \rightarrow [0,1]$ is a probability measure,
and $T: \Omega \hookleftarrow$ is a $\mathbb{B}(\Omega)$-$\mathbb{B}(\Omega)$-measurable $\mu$-preserving map,
i.e.  $\mu(T^{-1}(B)) = \mu(B)$  for all $B \in \mathbb{B}(\Omega)$.

The complexity of a system can be measured by considering a coarse-grained description of it provided by symbolic dynamics.
Given a finite partition ${\mathcal P} = \{P_0, P_1, \ldots, P_l\} \subset \mathbb{B}(\Omega)$ of $\Omega$ 
(below we consider only partitions ${\mathcal P} \subset \mathbb{B}(\Omega)$ without mentioning this explicitly), 
one assigns to each set $P_a \in {\mathcal P}$ the symbol $a$ from the alphabet $A=\{0, 1, \ldots, l\}$.
Similarly, the $n$-letter word $a_{0} a_{1} \ldots a_{n-1}$ is associated with the set $P_{a_{0} a_{1} \ldots a_{n-1}}$ defined by
\begin {equation}\label{RefinedPartitionElement_eq}
    P_{a_{0} a_{1} \ldots a_{n-1}} = P_{a_{0}} \cap T^{-1}(P_{a_{1}}) \cap \ldots \cap T^{-\circ(n-1)}(P_{a_{n-1}}).
\end {equation}
Then the collection
\begin {equation}\label{n_partition}
   {\mathcal P}_n = \{P_{a_{0}a_{1} \ldots a_{n-1}} \,\mid\, a_{0}, a_{1}, \ldots, a_{n-1} \in A\}
\end {equation}
forms a partition of $\Omega$ as well.
The {\it Shannon entropy}, the {\it entropy rate} and the {\it Kolmogorov-Sinai (KS) entropy} are respectively defined by (we use the convention that $0\ln 0:=0$)
\begin {gather*}
    H({\mathcal P}) = - \sum_{P \in {\mathcal P}} \mu(P) \ln \mu(P),\\
    h_\mu(T, {\mathcal P}) = \lim_{n \to \infty} \mleft( H({\mathcal P}_{n+1}) - H({\mathcal P}_n) \mright) = \lim_{n \to \infty} \frac{H({\mathcal P}_n)}{n},\\
    h_\mu(T) = \sup_{{\mathcal P}\text{ finite partition}}  h_\mu(T, {\mathcal P}).
\end {gather*}
The latter quantity provides a theoretical measure of complexity for a dynamical system.
In general, the determination of the KS entropy is complicated, 
thus the estimation of the KS entropy (from real-world data as well) is of interest.

\subsection{Ordinal patterns, permutation entropy and sorting entropy}\label{sectOrdinal}
Let us first recall the definitions of ordinal patterns and ordinal partitions.
For $d \in \mathbb{N}$ denote the set of permutations of $\lbrace 0, 1, \ldots , d\rbrace$ by $\Pi_d$.
\begin {definition}\label{OrdPatternDef}
    We say that a real vector $(x_0, x_1, \ldots, x_d)\in {\mathbb R}^{d+1}$ has the {\it ordinal pattern $\pi = (r_0, r_1,\ldots, r_d) \in \Pi_d$ } of order $d$ if
    \begin {equation*}
	x_{r_0} \geq x_{r_1} \geq \ldots \geq x_{r_d}
    \end {equation*}
    and
    \begin {equation*}
	r_{l-1} > r_{l} \text{ for } x_{r_{l-1}} = x_{r_{l}}
	\text{.}
     \end {equation*}
\end {definition}
\begin{definition}\label{OrdPartitionDef}
  For $N\in {\mathbb N}$, let ${\mathbf X}=(X_1,X_2,\ldots ,X_N)$ be an ${\mathbb R}$-valued random vector on $(\Omega,{\mathbb B}(\Omega))$.
  Then for $d\in {\mathbb N}$ the partition
  \begin{align*}
	{\mathcal P}^{\mathbf X}(d)=\{P_{(\pi_1,\pi_2,\ldots ,\pi_N)}\,\mid\,\pi_i\in\Pi_d\text{ for }i=1,2,\ldots ,N\}
  \end{align*}
  \hspace*{2.2mm}with
  \begin{align*}
      \hspace{2.2mm}
      P_{(\pi_1,\pi_2,\ldots ,\pi_N)}=\{\omega\in\Omega\,\mid\,(X_i(T^{\circ d}(\omega)),X_i(T^{\circ {d-1}}(\omega)),\ldots ,X_i(T(\omega)),X_i(\omega))\hspace*{4mm}\\
      \text{has the ordinal pattern }\pi_i\text{ for }i=1,2,\ldots ,N\}\hspace{4mm}
  \end{align*}
  is called the {\it ordinal partition of order $d$} with respect to $T$ and ${\mathbf X}$.
\end{definition}

The {\it permutation entropy of order $d$ (with respect to ${\mathbf X}$)} and the {\it sorting entropy of order $d$ (with respect to ${\mathbf X}$)}, 
being ordinal-based complexity measures for time series, are respectively given by
\begin {equation*}
     h_\mu ^{\mathbf X}(T, d) = \frac{1}{d}H({\mathcal P}^{\mathbf X}(d)),
\end {equation*}
\begin {equation*}
     h_{\mu, \triangle}^{\mathbf X}(T, d) = H({\mathcal P}^{\mathbf X}(d+1)) - H({\mathcal P}^{\mathbf X}(d))
\end {equation*}
(note that the original definitions in \cite{BandtPompe2002} were given for the case $\Omega \subseteq {\mathbb R}$ and ${\mathbf X} = \mathrm{id}$, where $\mathrm{id}$ is the identity map).
The permutation entropy is often defined just as $H({\mathcal P}^{\mathbf X}(d))$, but for us it is more convenient to use the definition above.
The sorting entropy represents the increase of diversity of ordinal patterns as the order $d$ increases by one.
To see the physical meaning of the permutation entropy let us rewrite it in the explicit form.
Given $\Pi_d^N = \{{\boldsymbol \pi} = (\pi_1,\pi_2,\ldots,\pi_N) \mid \pi_1,\pi_2,\ldots ,\pi_N \in \Pi_d\}$, we have
\begin {equation*} 
    h_{\mu}^{\mathbf X}(T, d) = - \frac{1}{d}\sum_{{\boldsymbol \pi} \in \Pi_d^N} \mu(P_{\boldsymbol \pi}) \ln \mu(P_{\boldsymbol \pi}),
\end {equation*}
that is the permutation entropy characterizes the diversity of ordinal patterns ${\boldsymbol \pi}$ divided by the order $d$.

In applications permutation and sorting entropy of order $d$ can be estimated from a finite orbit of a dynamical system with certain properties.
Simple and natural estimators are the empirical permutation entropy and the empirical sorting entropy, respectively. 
They are based on estimating $\mu(P_{(\pi_1,\pi_2,\ldots ,\pi_N)})$ by the empirical probabilities of observing $(\pi_1,\pi_2,\ldots ,\pi_N)$ in the time series generated by ${\mathbf X}$.

Finally, recall that the permutation and sorting entropy are related to the KS entropy.
For the case $\Omega \subseteq {\mathbb R}$, $T$ being a piecewise strictly-monotone interval map and ${\mathbf X} = \mathrm{id}$, Bandt et al. \cite{BandtKellerPompe2002} proved that:
\begin {equation*}
    h_\mu(T) = \lim_{d\to\infty} \frac{1}{d}H({\mathcal P}^\mathrm{id}(d)).
\end {equation*}
Keller and Sinn \cite{KellerSinn2009,KellerSinn2010, Keller2012} showed that in many cases (see \cite{AntonioukKellerMaksymenko2013} for recent results) it holds:
\begin {equation}\label{ordinalRepresentationKS}
    h_\mu(T) =\! \lim_{d\to\infty}h_\mu(T, {\mathcal P}^{\mathbf X}(d)) = \lim_{d\to\infty}\lim_{n\to\infty} \mleft(H({\mathcal P}^{\mathbf X}(d)_{n+1}) - H({\mathcal P}^{\mathbf X}(d)_n)\mright).
\end {equation}
Note that if \eqref{ordinalRepresentationKS} holds, then the permutation entropy and the sorting entropy for $d$ tending to infinity provide upper bounds for the KS entropy \cite{Keller2012}: 
\begin {equation*}
    \varlimsup_{d\to\infty} h_\mu ^{\mathbf X}(T, d) \geq h_\mu(T),\;\;\;\; \varlimsup_{d\to\infty} h_{\mu, \triangle}^{\mathbf X}(T, d) \geq h_\mu(T).
\end {equation*}
One may ask whether it is possible to get a better ordinal-based estimator of the KS entropy using the representation \eqref{ordinalRepresentationKS}.
This question is discussed in the next section.

\section{Conditional entropy of ordinal patterns and its relation to the Kolmogorov-Sinai entropy} \label{sectionCEofOPintro}
The {\it conditional entropy of ordinal patterns of order $d$} is defined by
\begin {equation}\label{CEofOPdefinition}
  h_{\mu, \text{cond}}^{\mathbf X}(T, d) = H({\mathcal P}^{\mathbf X}(d)_2) - H({\mathcal P}^{\mathbf X}(d)).
\end {equation}
It is the first element of the sequence 
\begin {equation*}
  \mleft( \mleft(H({\mathcal P}^{\mathbf X}(d)_{n+1}) - H({\mathcal P}^{\mathbf X}(d)_n)\mright) \mright)_{n \in {\mathbb N}}, 
\end {equation*}
which provides the ordinal representation \eqref{ordinalRepresentationKS} of the KS entropy as both $n$ and $d$ tend to infinity.
For brevity we refer to $h_{\mu, \text{cond}}^{\mathbf X}(T, d)$ as the `conditional entropy' when no confusion can arise.

To see the physical meaning of the conditional entropy recall that the entropies of the partitions ${\mathcal P}^{\mathbf X}(d)$ and ${\mathcal P}^{\mathbf X}(d)_2$ are given by
\begin {equation*}
      H({\mathcal P}^{\mathbf X}(d))   = -\sum_{{\boldsymbol \pi} \in \Pi_d^N} \mu(P_{\boldsymbol \pi}) \ln \mu(P_{\boldsymbol \pi}),
\end {equation*}
\begin {equation*}
      H({\mathcal P}^{\mathbf X}(d)_2) = -\sum_{{\boldsymbol \pi} \in \Pi_d^N} \sum_{{\boldsymbol \xi} \in \Pi_d^N} \mu(P_{\boldsymbol \pi} \cap T^{-1}(P_{\boldsymbol \xi})) 
													    \ln \mu(P_{\boldsymbol \pi} \cap T^{-1}(P_{\boldsymbol \xi})).
\end {equation*}
Then we can rewrite the conditional entropy \eqref{CEofOPdefinition} as
\begin {equation*} 
    h_{\mu, \text{cond}}^{\mathbf X}(T, d) = -\sum_{{\boldsymbol \pi} \in \Pi_d^N} \sum_{{\boldsymbol \xi} \in \Pi_d^N} \mu(P_{\boldsymbol \pi} \cap T^{-1}(P_{\boldsymbol \xi})) 
													      \ln \frac{\mu(P_{\boldsymbol \pi} \cap T^{-1}(P_{\boldsymbol \xi}))}{\mu(P_{\boldsymbol \pi})}.
\end {equation*}
If $\omega \in P_{\boldsymbol \pi} \cap T^{-1}(P_{\boldsymbol \xi})$ for some $P_{\boldsymbol \pi}, P_{\boldsymbol \xi} \in {\mathcal P}^{\mathbf X}(d)$ with 
${\boldsymbol \pi} = (\pi_1,\pi_2,\ldots,\pi_N)$ and ${\boldsymbol \xi} = (\xi_1,\xi_2,\ldots,\xi_N)$, 
then we say that in $\omega$ the ordinal patterns $\xi_1,\xi_2,\ldots,\xi_N$ are {\it successors} of the ordinal patterns $\pi_1,\pi_2,\ldots,\pi_N$, respectively.
The conditional entropy characterizes the diversity of successors of given ordinal patterns ${\boldsymbol \pi}$, 
whereas the permutation entropy characterizes the diversity of ordinal patterns ${\boldsymbol \pi}$ themselves.

In the rest of the section we discuss the relationship between the conditional entropy of ordinal patterns and the KS entropy.

\subsection{Relationship in the general case}\label{sectCEofOPandKSgeneral}
Statements {\it(i)} and {\it(ii)} of the following theorem imply that under the given assumptions  
the conditional entropy of ordinal patterns bounds the KS entropy better than the sorting entropy and the permutation entropy, respectively.

\begin{theorem}\label{CEofOP_andOrdinalEntr}
   Let $\mleft( \Omega, \mathbb{B}(\Omega), \mu, T \mright)$ be a measure-preserving dynamical system, 
   ${\mathbf X}$ be a random vector on $(\Omega,{\mathbb B}(\Omega))$ such that \eqref{ordinalRepresentationKS} is satisfied. Then it holds
   \begin {flalign*}
	  &\text{{\it(i)}\hphantom{{\it(ii)}}} 
			h_{\mu}(T) \leq \varlimsup_{d \to \infty}h_{\mu, \text{cond}}^{\mathbf X}(T, d) \leq \varlimsup_{d \to \infty}h_{\mu, \triangle}^{\mathbf X}(T, d).&\\
	  &\text{{\it(ii)}\hphantom{{\it(i)}}} 
			 \text{Moreover, if for some } d_0 \in {\mathbb N} \text{ it holds}&
   \end {flalign*}
   \begin {equation}\label{PEdecreases}
	  h_\mu ^{\mathbf X}(T, d) \geq h_\mu ^{\mathbf X}(T, d+1) \text{ for all } d \geq d_0,
   \end {equation}
   \begin {flalign*}
	  &\text{\hphantom{{\it(i)}}\hphantom{{\it(ii)}}or the limit of the sorting entropy}&
   \end{flalign*} 
   \begin {equation}\label{SElimitExists}
	  \lim_{d \to \infty} h_{\mu, \triangle}^{\mathbf X}(T, d)\text{ exists},
   \end {equation}
   \begin {flalign*}
	  &\text{\hphantom{{\it(i)}}\hphantom{{\it(ii)}}then it holds:}&
    \end{flalign*}
    \begin {equation}\label{EntropiesLimit_ineq}
          h_{\mu}(T) \leq \varlimsup_{d \to \infty}h_{\mu, \text{cond}}^{\mathbf X}(T, d) \leq \varlimsup_{d \to \infty}h_{\mu, \triangle}^{\mathbf X}(T, d) \leq 
											       \varlimsup_{d \to \infty}h_\mu^{\mathbf X}(T, d).
   \end {equation}
\end {theorem}

The proof is given in Subsection~\ref{CEofOP_andOrdinalEntr_ProofSect}.
Note that both statements of Theorem~\ref{CEofOP_andOrdinalEntr} remain correct if one replaces the upper limits by the lower limits.

As a consequence of Theorem~\ref{CEofOP_andOrdinalEntr} we get the following result.
\begin{corollary}\label{CEofOPandKS_PE}
    If the assumption \eqref{ordinalRepresentationKS} and either of the assumptions \eqref{PEdecreases} or \eqref{SElimitExists} are satisfied, then $h_{\mu}(T) = \varlimsup_{d \to \infty}h_\mu^{\mathbf X}(T, d)$ yields
    \begin {equation*}
          h_{\mu}(T) = \varlimsup_{d \to \infty}h_{\mu, \text{cond}}^{\mathbf X}(T, d).
    \end {equation*}
\end{corollary}

This sheds some light on the behavior of the conditional entropy for the logistic maps, described in Example~\ref{LogisticFamilyEx}.
Nevertheless, it is not clear whether the statements \eqref{PEdecreases} or \eqref{SElimitExists} hold, neither in the general case nor for the logistic maps.
Note that a sufficient condition for \eqref{SElimitExists} is the monotone decrease of the sorting entropy $h_{\mu, \triangle}^{\mathbf X}(T, d)$ with increasing $d$.
However, the sorting entropy and the permutation entropy do not necessarily decrease for all $d$.

\begin{example}\label{GoldenMeanEx}
  Consider the golden mean map $T_{gm}: [0,1]\hookleftarrow$ defined by 
  \begin {equation*}
      T_{gm}(\omega) = \begin{cases}
			  \varphi \omega,  	& 0 \leq \omega \leq \frac{1}{\varphi},\\
			  \varphi \omega - 1,	& \frac{1}{\varphi}  < \omega \leq 1,
		       \end{cases}
  \end {equation*}
  for $\varphi = (\sqrt{5} + 1)/2$ being the golden ratio.
  The map $T_{gm}$ preserves the measure $\mu_{gm}$ \cite{Choe2005} given by $\mu_{gm}(U) = \int_U p(\omega) d\omega$ for all $U \in \mathbb{B}\mleft([0,1]\mright)$ and for
  \begin {equation*}
      p(\omega) = \begin{cases}
			\frac{\varphi^3}{1 + \varphi^2},  	& 0 \leq \omega \leq \frac{1}{\varphi},\\
			\frac{\varphi^2}{1 + \varphi^2},	& \frac{1}{\varphi}  < \omega \leq 1.
		  \end{cases}
  \end {equation*}
  The values of permutation, sorting and conditional entropies for the dynamical system $\mleft([0,1],\mathbb{B}\mleft([0,1]\mright), \mu_{gm}, T_{gm} \mright)$ 
  estimated from the orbit of length $L = 4\cdot 10^6$ are shown in Figure~\ref{GoldenMeanEntropies}.
  Note that neither sorting nor permutation entropy is monotonically decreasing with increasing $d$.
  (The interesting fact that for all $d = 1,2,\ldots,9$ the conditional entropy and the KS entropy coincide is explained in Subsection~\ref{sectMarkov}.)

  \begin{figure}[!th]
	  \centering
	  \includegraphics[scale=0.48]{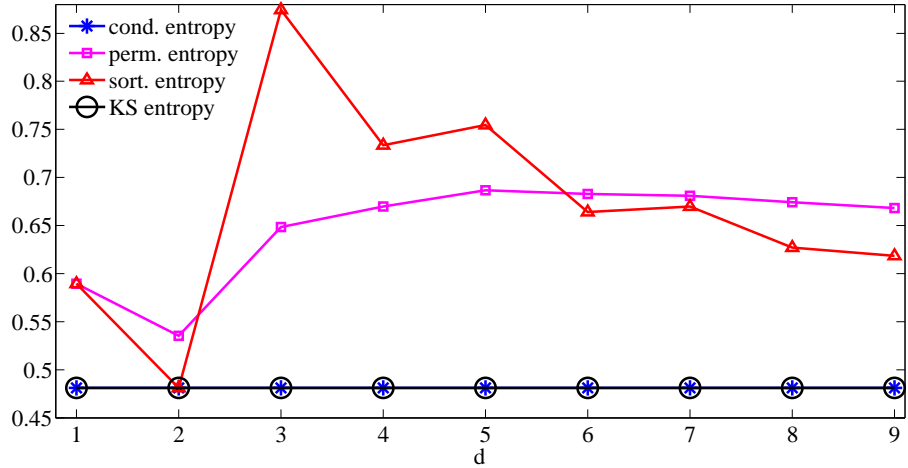}
          \caption{Conditional entropy, permutation entropy and sorting entropy in comparison with the KS entropy for the golden mean map}
      \label{GoldenMeanEntropies}
  \end{figure}
\end{example}

The question when $h_{\mu, \triangle}^{\mathbf X}(T, d)$ or $h_{\mu}^{\mathbf X}(T, d)$ decrease starting from some $d_0 \in \mathbb{N}$ is still open. 
For instance, for the logistic map with $r = 4$ our estimated values of permutation entropy and sorting entropy decrease starting from $d = 7$ and $d = 4$, respectively 
(see Figure~\ref{figureLogistic}).
However, at this point we do not have theoretical results in this direction.

\begin{figure}[!th]
      \centering
      \includegraphics[scale=0.48]{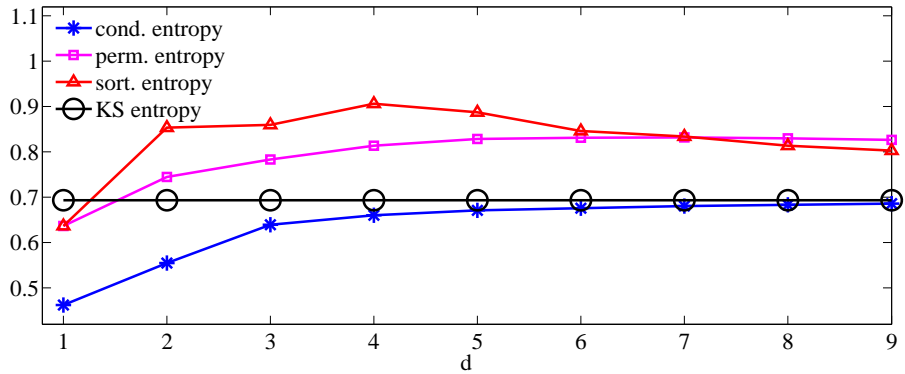}
      \caption{Empirical conditional entropy, permutation entropy and sorting entropy in comparison with the KS entropy for the logistic map with $r=4$}
      \label{figureLogistic}
\end{figure}

\subsection{Markov property of ordinal partition}\label{sectMarkovPropertyOfOP}
Computation of the KS entropy involves taking a supremum over all finite partitions and is unfeasible in the general case. 
A possible solution is provided by the properties given in Definitions \ref{GeneratingPartition} and \ref{MarkovPartition_measure}.
\begin {definition}\label{GeneratingPartition}
  A finite partition ${\mathcal G} = \{G_0, G_1, \ldots ,G_l\} \subset \mathbb{B}(\Omega)$ of $\Omega$ is said to be {\it generating} (under $T$)
  if, given $\mathbb{A}$ the sigma-algebra generated by the sets $T^{-\circ n}(G_i)$ with $i = 0, 1, \ldots, l$ and $n \in \mathbb{N}_0$, 
  for every $B \in \mathbb{B}(\Omega)$ there exists a set $A \in \mathbb{A}$ such that $\mu(A \bigtriangleup B) = 0$.
\end {definition}
\begin {definition}\label{MarkovPartition_measure}
    A finite partition ${\mathcal M} = \{M_0, M_1, \ldots, M_l\} \subset \mathbb{B}(\Omega) $ of $\Omega$ has the {\it Markov property} with respect to $T$ and $\mu$ if
    for all $i_0,i_1, \ldots, i_n \in \{0, 1, \ldots, l\}$ with $n \in \mathbb{N}$ and 
    $\mu\mleft(M_{i_{0}} \cap T^{-1}(M _{i_{1}}) \cap \ldots \cap T^{-\circ(n-1)}(M_{i_{n-1}})\mright) > 0$ it holds 
	\begin {equation}\label{MarkovPartitionDefinition} 
	      \frac{\mu\mleft(M_{i_{0}} \cap T^{-1}(M_{i_{1}}) \cap \ldots \cap T^{-\circ n}(M_{i_{n}})\mright)}
	           {\mu\mleft(M_{i_{0}} \cap T^{-1}(M_{i_{1}}) \cap \ldots \cap T^{-\circ(n-1)}(M_{i_{n-1}})\mright)} =
	      \frac{\mu\mleft(M_{i_{n-1}}\! \cap T^{-1}(M_{i_{n}})\mright)}{\mu(M_{i_{n-1}})}.
	\end {equation}
 \end {definition}
Originally in \cite{ParryWilliams1977} a partition with property \eqref{MarkovPartitionDefinition} was called Markov partition, 
but we use another term to avoid confusion with the topological notion of Markov partition.

By the Kolmogorov-Sinai theorem (for details we refer to \cite[Theorem~4.17]{Walters1982}), if ${\mathcal G}$ is a generating partition then it holds $h_\mu(T) = h_\mu(T, {\mathcal G})$.
Further, it is easy to show (see \cite[Observation 6.2.10]{Kitchens1998}) that for the partition ${\mathcal M}$ with the Markov property it holds
\begin {equation*}
    h_\mu(T, {\mathcal M}) = H({\mathcal M}_2) - H({\mathcal M}).
\end {equation*}
Therefore, if ${\mathcal M}$ is both generating and has the Markov property, then
\begin {equation*}
   h_\mu(T) = H({\mathcal M}_{2}) - H(\mathcal M).
\end {equation*}

From the last two statements it follows the sufficient condition for the coincidence between the conditional entropy and the KS entropy.
\begin{lemma}\label{MarkovPropertyOfOP_Th}
   Let $\mleft( \Omega, \mathbb{B}(\Omega), \mu, T \mright)$ be a measure-preserving dynamical system, 
   ${\mathbf X}$ be an ${\mathbb R}$-valued random vector on $(\Omega,{\mathbb B}(\Omega))$ such that \eqref{ordinalRepresentationKS} is satisfied.
   Then the following two statements hold:
   \begin{myindentpar}{0mm}
      \noindent{\it(i)}\hphantom{{\it(ii)}}If ${\mathcal P}^{\mathbf X}(d)$ has the Markov property for all $d \geq d_0$ then 
      \begin {equation*}
	  h_\mu(T) = \lim_{d \to \infty}h_{\mu, \text{cond}}^{\mathbf X}(T, d).
      \end {equation*}
      {\it(ii)}\hphantom{{\it(i)}}If ${\mathcal P}^{\mathbf X}(d)$ is generating and has the Markov property for some $d \in \mathbb{N}$ then
      \begin {equation}\label{MarkovOrdinalRepresentation}
	  h_\mu(T) = h_{\mu, \text{cond}}^{\mathbf X}(T, d).
      \end {equation}    
   \end{myindentpar}
\end{lemma} 
In general, it is complicated to verify that ordinal partitions are generating or have the Markov property; 
however in Subsection~\ref{sectMarkov} this is done for Markov shifts over two symbols.

\subsection{Markov shifts}\label{sectMarkov}
In this subsection we establish the equality of the conditional entropy of order $d$ and the KS entropy for the case of Markov shifts over two symbols.
First we recall the definition of the Markov shifts (see \cite[Section 6]{Kitchens1998} for details), 
then we impose a natural restriction on the observables, state the result and finally discuss some possible extensions.

\begin {definition}\label{MarkovShift_def}
  Let $Q = (q_{ij})$ be an $(l+1) \times (l+1)$ stochastic matrix and 
  $p = (p_0, p_1, \ldots, p_l)$ be a stationary probability vector of $Q$ with $p_0, p_1, \ldots, p_l > 0$.
  Then a {\it Markov shift} is the dynamical system $(A^\mathbb{N}, \mathbb{B}_\Pi(A^\mathbb{N}),m, \sigma)$, where
  \begin{itemize}
   \item $A^\mathbb{N}$ is the space of one-sided sequences over $A =\{0, 1, \ldots, l\}$,
   \item $\mathbb{B}_\Pi(A^\mathbb{N})$ is the Borel sigma-algebra generated by the cylinder sets $C_{a_0 a_1 \ldots a_{n-1}}$ given by
   \begin {equation*}
      C_{a_0 a_1 \ldots a_{n-1}} = \{(s_0, s_1, \ldots) \in A^\mathbb{N} \mid s_0 = a_0, s_1 = a_1, \ldots, s_{n-1} = a_{n-1}\},
   \end {equation*}
   \item $\sigma: A^\mathbb{N} \hookleftarrow$ such that $(\sigma s)_j = s_{j+1}$ for all $j \in \mathbb{N}_0$ and $s = (s_0, s_1, \ldots) \in A^\mathbb{N}$ is a {\it shift map},
   \item $m$ is a {\it Markov measure} on $A^\mathbb{N}$, defined on the cylinder sets $C_{a_0 a_1 \ldots a_{n-1}}$ by 
   \begin {equation*}
      m\mleft(C_{a_0 a_1 \ldots a_{n-1}}\mright) = p_{a_0}q_{a_0 a_1}q_{a_1 a_2}\cdots q_{a_{n-2} a_{n-1}}.
   \end {equation*}
  \end{itemize}
\end {definition}

In the particular case when $q_{0a} = q_{1a} = \ldots = q_{la} = p_a$ for all $a \in A$, the measure $m_B$ defined as follows is said to be a {\it Bernoulli measure}:
\begin {equation*}
   m_B\mleft(C_{a_0 a_1 \ldots a_{n-1}}\mright) = p_{a_0}p_{a_1}\cdots p_{a_{n-1}}.
\end {equation*}
The system $(A^\mathbb{N}, \mathbb{B}_\Pi(A^\mathbb{N}),m_B, \sigma)$ is then called a {\it Bernoulli shift}. 
We use this concept below for illustration purposes.

The natural order on the set of sequences is the lexicographic order defined as follows: 
for $r = (r_0, r_1, \ldots),\, s = (s_0, s_1, \ldots) \in A^\mathbb{N}$ the inequality $r \prec s$ holds iff 
$r_0 < s_0$ or there exists some $k \in\mathbb{N}$ with $r_i = s_i$ for $i = 0, 1, \ldots, k-1$ and $r_k < s_k$.
However, we prefer to be consistent in using the concept of ordinal patterns 
and to keep working with the usual order on observables instead of considering particular orders on different spaces.
Thus, to introduce the ordinal partition for Markov shifts, we impose a restriction on the observables on $A^\mathbb{N}$.

\begin {definition}
  Let us say that the observable $X: A^\mathbb{N} \to \mathbb{R}$ is {\it lexicographic-like} if for almost all $s \in A^\mathbb{N}$ it is injective 
  and if for all $s \in A^\mathbb{N}$ and $m, n \in \mathbb{N}_0$ 
  the inequality $X(\sigma^{\circ m}s) \leq X(\sigma^{\circ n}s)$ holds iff $\sigma^{\circ m}s \preceq \sigma^{\circ n}s$.
\end {definition}
In other words, the fact that an observable $X$ is lexicographic-like means that an $X$ induces the natural order relation on $A^\mathbb{N}$.
A simple example of such $X$ is provided by considering $s \in A^\mathbb{N}$ as $(l+1)$-expansions of a number in $[0,1)$:
\begin {equation*}
  X_\text{expans}((s_0, s_1, \ldots)) = \sum_{j=0}^\infty \mleft(\frac{1}{l+1}\mright)^{j+1} s_j.
\end {equation*}
Note that since a lexicographic-like $X$ is injective for almost all $s \in A^\mathbb{N}$, 
it provides the ordinal representation \eqref{ordinalRepresentationKS} of the KS entropy (see \cite{Keller2012}).

Recall that the dynamical system $\left( \Omega, \mathbb{B}(\Omega), \mu, T \right)$ is {\it ergodic} 
if for every $B \in \mathbb{B}(\Omega)$ with $T^{-1}B = B$ it holds either $\mu(B)= 0$ or $\mu(B)= 1$.
In Subsection~\ref{sectionMarkov2symbols} we prove the following statement.
\begin {lemma}\label{theorem2letters}
  Let $(\{0, 1\}^\mathbb{N}, \mathbb{B}_\Pi({0, 1}^\mathbb{N}),m, \sigma)$ be an ergodic Markov shift over two symbols.
  If the random variable $X$ is lexicographic-like, then the ordinal partition ${\mathcal P}^X(d)$ is generating and has the Markov property for all $d \in \mathbb{N}$.
\end {lemma}
As a direct consequence of Lemmas~\ref{MarkovPropertyOfOP_Th} and~\ref{theorem2letters} we obtain the following.
\begin{theorem}\label{MarkovShift_cor}
  Under the assumptions of Lemma~\ref{theorem2letters} for all $d\in \mathbb{N}$ it holds 
  \begin {equation*}
    h_m(\sigma) = h_{m, \text{cond}}^{X}(\sigma, d).
  \end {equation*}
\end{theorem} 

\begin{example}\label{Bernoulli2Ex}
  Figure~\ref{figureBernulli1} illustrates Theorem~\ref{MarkovShift_cor} for the Bernoulli shift over two symbols with $m_B(C_0) = 0.663$, $m_B(C_1) = 0.337$.
  For all $d = 1,2,\ldots,9$ the empirical conditional entropy computed from an orbit of length $L = 3.6 \cdot 10^6$ nearly coincides with the theoretical KS entropy $h_{m_B}(\sigma)$.
  Meanwhile, the empirical permutation entropy differs from the KS entropy significantly.
  \begin{figure}[h]
      \centering
      \includegraphics[scale=0.48]{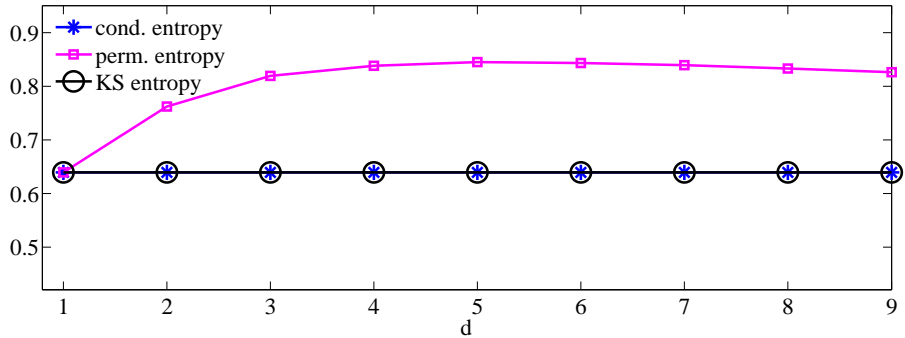}
      \caption{Empirical conditional entropy and permutation entropy in comparison with the KS entropy for the Bernoulli shift over two symbols}
      \label{figureBernulli1}
  \end{figure}
\end{example}

The result established in Theorem~\ref{MarkovShift_cor} naturally extends to the class of maps that are order-isomorphic to an ergodic Markov shift over two symbols
(the concept of order-isomorphism is introduced in \cite{Amigo2010}).
An example of a map being order-isomorphic to a Markov shift is the golden mean map considered in Example~\ref{GoldenMeanEx}.
It explains the coincidence of the conditional entropy and the KS entropy in Figure~\ref{GoldenMeanEntropies}.
Note that the logistic map for $r = 4$ is isomorphic, but not order-isomorphic to an ergodic Markov shift over two symbols \cite[Subsection 3.4.1]{Amigo2010}. 
Therefore in the case of the logistic map with $r = 4$ the conditional entropy for finite $d$ does not coincide with the KS entropy (see Figure~\ref{figureLogistic}).

Theorem~\ref{MarkovShift_cor} cannot be extended to Markov shifts over a general alphabet.
One can rigorously show that for the Bernoulli shifts over more than two symbols, ${\mathcal P}^X(d)$ does not have the Markov property.

\begin{example}\label{Bernoulli34Ex}
  Figure~\ref{figureBernulli23} represents estimated values of the empirical conditional and permutation entropies for the Bernoulli shifts over three and four symbols.
  Although these shifts have the same KS entropy as the shift in Figure \ref{figureBernulli1}, their conditional entropies differ significantly.

 \begin{figure}[h]
      \centering
      \includegraphics[scale=0.48]{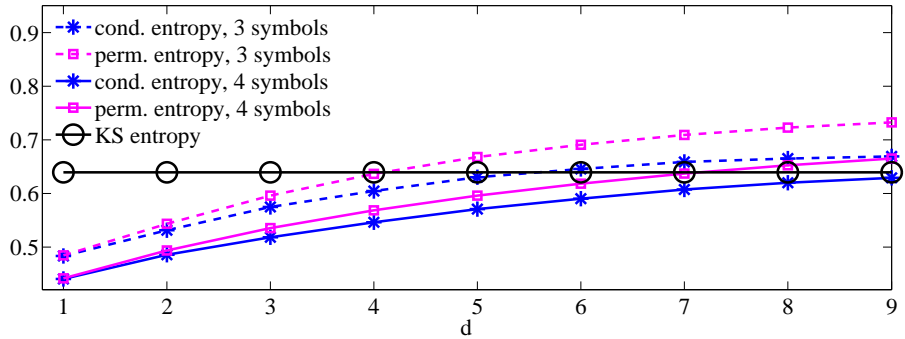}
      \caption{Empirical conditional entropy and permutation entropy in comparison with the KS entropy for Bernoulli shifts over three and four symbols}
      \label{figureBernulli23}
  \end{figure}
\end{example}

\subsection{Periodic case}\label{sectPeriodicCase}
Here we relate the conditional entropy to the KS entropy in the case of periodic dynamics. 
By periodic dynamical system we mean a system such that the set of periodic points has measure $1$. 
Though it is well known that the KS entropy of a periodic dynamical system is equal to zero,  
the permutation entropy of order $d$ can be arbitrarily large in this case and thus does not provide a reliable estimate for the KS entropy. 
In Theorem~\ref{AllPeriodicPointsTh} we show that 
the conditional entropy of a periodic dynamical system is equal to the KS entropy starting from some finite order $d$, which advantages the conditional entropy over the permutation entropy.

\begin{theorem}\label{AllPeriodicPointsTh}
  Let $\mleft( \Omega, \mathbb{B}(\Omega), \mu, T \mright)$ be a measure-preserving dynamical system.
  Suppose that the set of periodic points of $\Omega$ with period not exceeding $k \in \mathbb{N}$ has measure $1$, then for all $d \in \mathbb{N}$ with
  $d \geq k$ it holds
\begin {equation*}
  h_{\mu, \text{cond}}^{\mathbf X}(T, d) = h_\mu(T) = 0.
\end {equation*}
\end{theorem} 
The proof is given in Subsection~\ref{AllPeriodicPointsTh_ProofSect}.
\begin{example}\label{RotationMapsEx}
    In order to illustrate the behavior of permutation and conditional entropies of periodic dynamical systems, 
    consider the rotation maps $g_\alpha(\omega) = (\omega + \alpha) \mod 1$ on the interval $[0,1]$ with the Lebesgue measure $\lambda$.
    Let $\alpha$ be rational, then $g_\alpha(\omega)$ provides a periodic behavior and it holds $h_\lambda ^\mathrm{id}(g_\alpha) = 0$.
    Figure~\ref{figureRotation} illustrates conditional and permutation entropies for the rotation maps for $d = 4$ and $d = 8$ for $\alpha$ varying with step 0.001. 
    For both values of $d$ the conditional entropy is more close to zero than the permutation entropy
    since periodic orbits provide various ordinal patterns, but most of them have one and the same successor.
    Note that for those values of $\alpha$ forcing periods shorter than $d$ (for instance for $\alpha = 0.25$ all $\omega \in [0,1]$ have period $4$) 
    it holds $h_{\lambda, \text{cond}}^\mathrm{id}(g_\alpha, d) = h_{\lambda}(g_\alpha) = 0$ as provided by Theorem~\ref{AllPeriodicPointsTh}.

    \begin{figure}[h]
      \begin{minipage}[h]{0.49\hsize}
	  \centering
	  \includegraphics[scale=0.45]{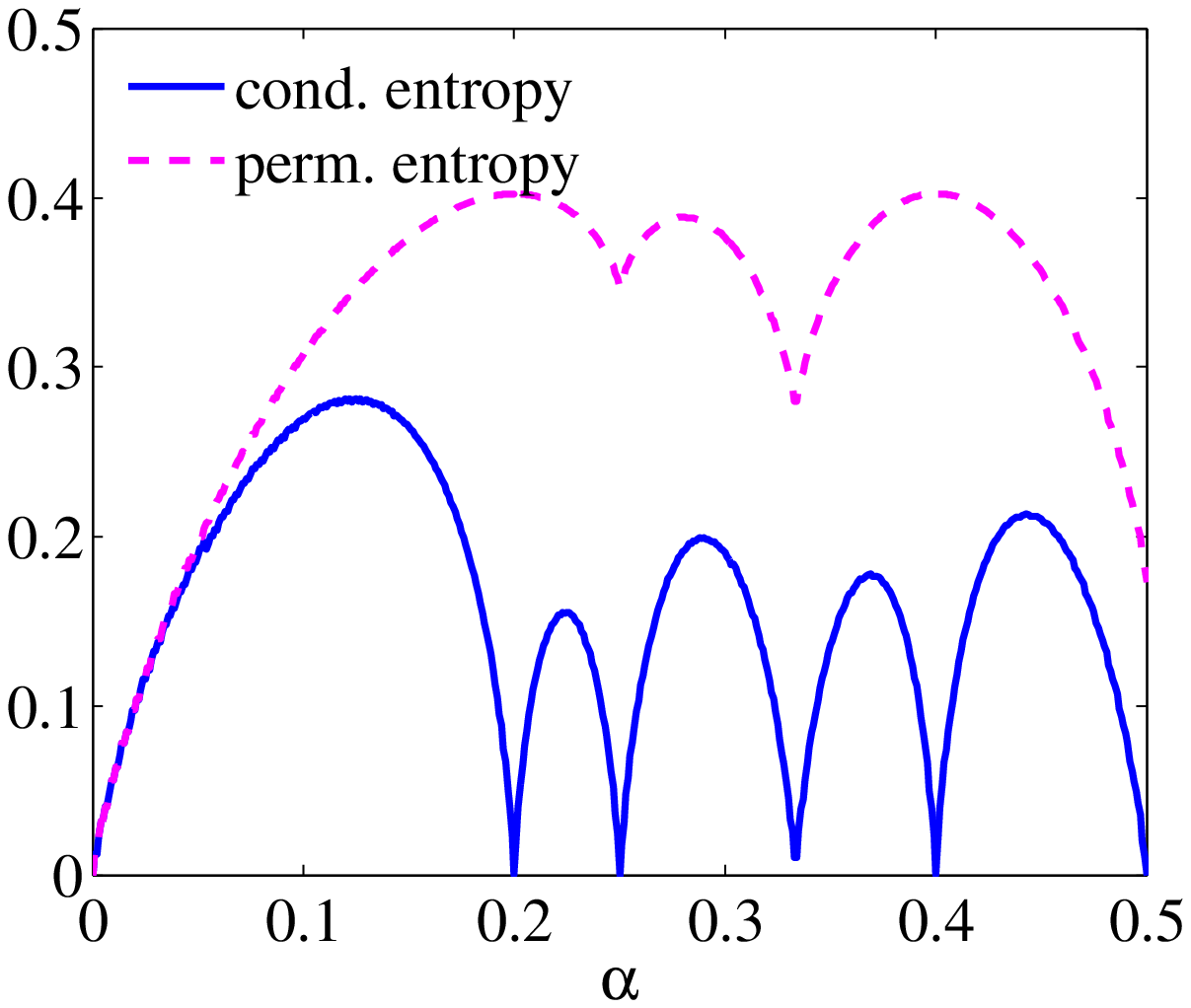}
	  
	  (a)
      \end{minipage}
      \begin{minipage}[h]{0.49\hsize}
	  \centering
	  \includegraphics[scale=0.45]{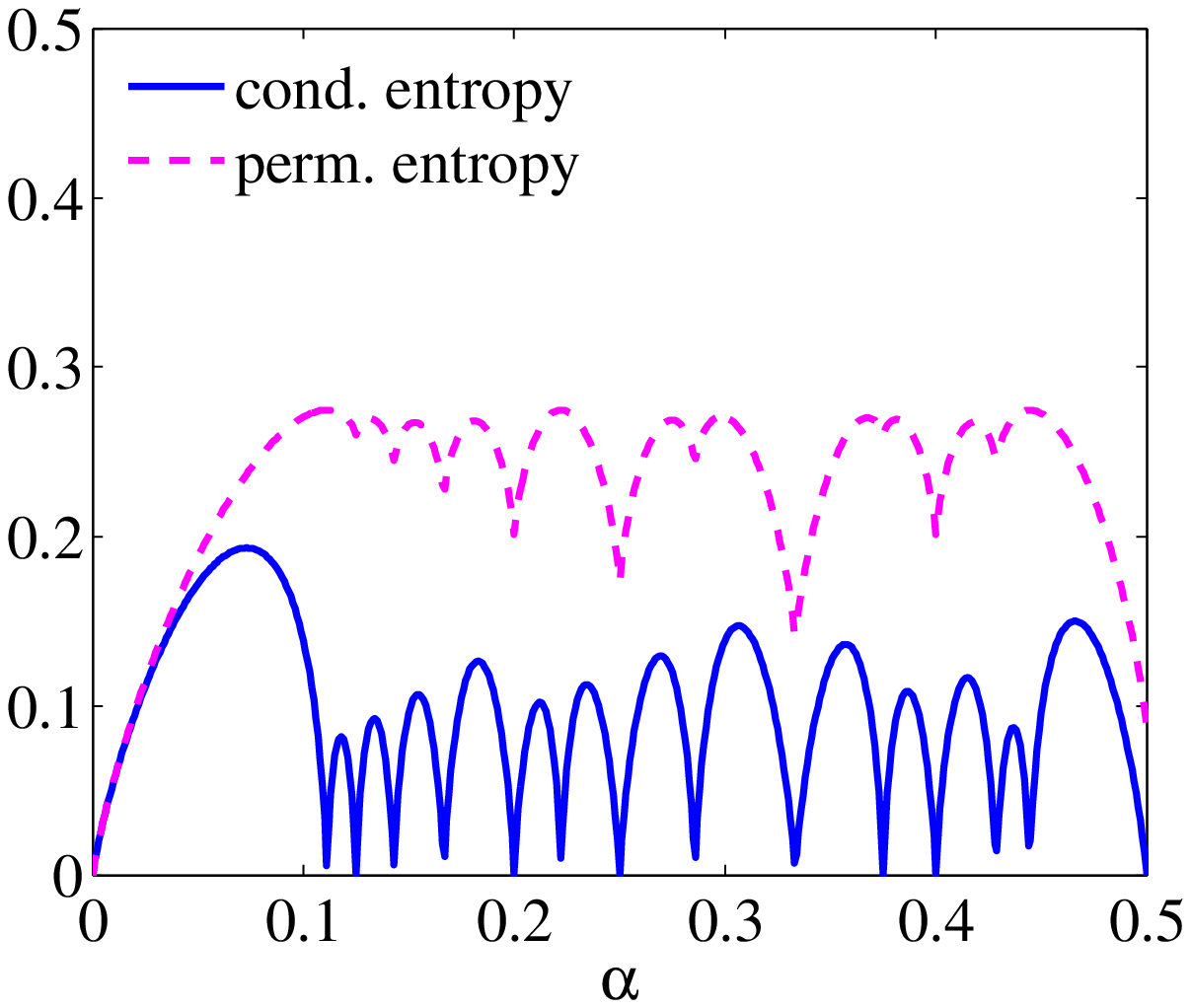}
	  
	  (b)
      \end{minipage}
      \caption{Conditional and permutation entropy for rotation maps for $d = 4$ (a) and $d = 8$ (b)}
      \label{figureRotation}
    \end{figure}
\end{example}

\subsection{Relationship between permutation entropy and Kolmogorov-Sinai entropy}\label{sectOrdinalKS}
We finish this section by giving several reasons for why the permutation entropy of finite $d$ does not provide an appropriate estimation of the KS entropy (even if $ h_\mu(T) = h_\mu ^{\mathbf X}(T)$).
First, the permutation entropy of order $d$ converges to $h_\mu ^{\mathbf X}(T)$ rather slowly \cite{BandtKellerPompe2002}.
Second, the permutation entropy of order $d$ is bounded from above, which means that for a relatively small $d$ a relatively large KS entropy cannot be correctly estimated.
Indeed, given an ${\mathbb R}$-valued random vector ${\mathbf X}=(X_1,X_2,\ldots ,X_N)$ on $(\Omega,{\mathbb B}(\Omega))$, for all $d \in {\mathbb N}$ it holds
    \begin {equation}\label{PEboundary_ineq}
	h_\mu ^{\bf X}(T, d) \leq \frac{N\ln\mleft((d+1)!\mright)}{d}.
    \end {equation}
To see this recall that there are $(d+1)!$ different ordinal patterns of order $d$.
Therefore by general properties of the Shannon entropy we have
    \begin {equation*}
	H({\mathcal P}^{\bf X}(d)) \leq \ln((d+1)!)^N = N \ln\mleft((d+1)!\mright)
    \end {equation*}
and inequality \eqref{PEboundary_ineq} becomes obvious.

Finally, as we have already mentioned in Subsection~\ref{sectPeriodicCase}, the permutation entropy of order $d$ can be arbitrarily large for simple systems.
In particular, for any given $d \in {\mathbb N}$ one can construct a periodic dynamical system such that the permutation entropy of order $d$ reaches the maximal possible value 
provided by inequality \eqref{PEboundary_ineq}, while $h_\mu^{\bf X}(T) = h_\mu(T) = 0$. 
\begin{example}\label{HighPEandLowKS_ex}
 Consider a map $T_{p}(\omega)$, defined as follows
      \begin {equation*}
 T_{p}(\omega) =  \begin{cases}
			    \omega + \frac{5}{6},  		& 0 	   \leq \omega \leq \frac{1}{6},\\
			    \omega - \frac{1}{6}, 		& \frac{1}{6} < \omega \leq \frac{2}{6},\\
			    \omega + \frac{1}{6}, 		& \frac{2}{6} < \omega \leq \frac{4}{6},\\
			    \omega - \frac{3}{6}, 		& \frac{4}{6} < \omega \leq 1,
		  \end{cases}
      \end {equation*}
  on the interval $\Omega = [0,1]$ with the Lebesgue measure $\lambda$. 
  All points $\omega \in \Omega$ are periodic with period $6$, as one can easily check, therefore $h_\lambda(T_{p}) = 0$.
  However, all ordinal patterns of order $d = 1, 2$ occur with equal frequency $\frac{1}{(d+1)!}$, 
  which provides $h_\lambda^\mathrm{id}(T_{p}, d) = \frac{1}{d}\ln\left((d+1)!\right)$ for  $d = 1, 2$ (cf. \eqref{PEboundary_ineq}).
\end{example}

\section{Interrelationship between conditional entropy of ordinal patterns, permutation and sorting entropy}\label{sectionCEofOP_PE_SE}
In this section we consider the relationship between the conditional entropy of order $d$, 
the permutation entropy $h_\mu ^{\mathbf X}(T, d)$ and the sorting entropy $h_{\mu, \triangle}^{\mathbf X}(T, d)$.
Besides being interesting in its own right, this relationship is used to prove Theorem~\ref{CEofOP_andOrdinalEntr}. 
\begin {lemma}\label{SEandCEofOP}
   Let $\mleft( \Omega, \mathbb{B}(\Omega), \mu, T \mright)$ be a measure-preserving dynamical system. 
   Then for all $d \in \mathbb{N}$ it holds
   \begin {equation}\label{SEandCEofOPineq}
       h_{\mu, \text{cond}}^{\mathbf X}(T, d) \leq h_{\mu, \triangle}^{\mathbf X}(T, d).
   \end {equation}
   Moreover, if for some $d_0 \in \mathbb{N}$ it holds $h_\mu ^{\mathbf X}(T, d_0+1) \leq h_\mu ^{\mathbf X}(T, d_0)$, then we get
    \begin {equation}\label{SEPEandCEofOPineq}
       h_{\mu, \text{cond}}^{\mathbf X}(T, d_0) \leq h_{\mu, \triangle}^{\mathbf X}(T, d_0) \leq h_\mu ^{\mathbf X}(T, d_0).
   \end {equation}
\end {lemma}
\begin {proof}
    It can easily be shown (for details see \cite{KellerSinn2010}) that for all $d \in {\mathbb N}$ it holds
    \begin {equation*}
	H({\mathcal P}^{\mathbf X}(d)_2) \leq H({\mathcal P}^{\mathbf X}(d+1)),
    \end {equation*}
    which implies
    \begin {align*}
	h_{\mu, \text{cond}}^{\mathbf X}(T, d) = H({\mathcal P}^{\mathbf X}(d)_2) - H({\mathcal P}^{\mathbf X}(d)) &\leq H({\mathcal P}^{\mathbf X}(d+1)) - H({\mathcal P}^{\mathbf X}(d))\\
											       &= h_{\mu, \triangle}^{\mathbf X}(T, d),
    \end {align*}
    and the proof of \eqref{SEandCEofOPineq} is complete.

    If $h_\mu ^{\mathbf X}(T, d_0+1) \leq h_\mu ^{\mathbf X}(T, d_0)$ for some $d_0 \in \mathbb{N}$ then we have
    \begin {equation*}
	d_0 H({\mathcal P}^{\mathbf X}(d_0+1)) \leq (d_0+1) H({\mathcal P}^{\mathbf X}(d_0)).
    \end {equation*}    
    Consequently, it holds
    \begin {equation*}
	d_0 (H({\mathcal P}^{\mathbf X}(d_0+1)) - H({\mathcal P}^{\mathbf X}(d_0))) \leq H({\mathcal P}^{\mathbf X}(d_0)),
    \end {equation*}    
    which establishes \eqref{SEPEandCEofOPineq}.
\end {proof}

By Lemma~\ref{SEandCEofOP} we have that the conditional entropy under certain assumption is not greater than the permutation entropy 
and that in the general case the conditional entropy is not greater than the sorting entropy.
Moreover, one can show that in the strong-mixing case the conditional entropy and the sorting entropy asymptotically approach each other. 
To see this recall that the map $T$ is said to be {\it strong-mixing} if for every $A,B \in \mathbb{B}(\Omega)$ it holds
    \begin{equation*}
	\lim_{n \rightarrow \infty}\mu(T^{-\circ n}A \cap B)=\mu(A)\mu(B).
    \end{equation*}
According to \cite{UnakafovaUnakafovKeller2013}, if $\Omega$ is an interval in $\mathbb{R}$ and $T$ is strong-mixing then it holds
    \begin {equation*}
	\lim_{d\to\infty} \mleft( H({\mathcal P}^{\mathrm{id}}(d+1))- H({\mathcal P}^{\mathrm{id}}(d)_2) \mright) = 0.
    \end {equation*}
Together with Lemma~\ref{SEandCEofOP} this implies the following statement.
\begin {corollary}\label{SEandCEofOPeq_propos}
   Let $\mleft( \Omega, \mathbb{B}(\Omega), \mu, T \mright)$ be a measure-preserving dynamical system, where $\Omega$ is an interval in $\mathbb{R}$ and $T$ is strong-mixing. Then  
   \begin {equation*}
   	\lim_{d\to\infty}\mleft(h_{\mu, \triangle}^{\mathrm{id}}(T, d) - h_{\mu, \text{cond}}^{\mathrm{id}}(T, d) \mright) = 0.
    \end {equation*}
\end {corollary}

\section{Conclusions}\label{sectionConclusion}
As we have discussed, the conditional entropy of ordinal patterns has rather good properties.
Our theoretical results and numerical experiments show that in many cases the conditional entropy provides a reliable estimation of the KS entropy.
In this regard it is important to note that the conditional entropy is computationally simple: it has the same computational complexity as the permutation entropy.
(The algorithm for fast computing the conditional entropy of ordinal patterns, based on ideas presented in \cite{UnakafovaKeller2013}, will be discussed elsewhere.)

Meanwhile, some questions concerning the conditional entropy of ordinal patterns remain open.
In particular, possible directions of a future work are to find dynamical systems having one of the following properties:
\begin{enumerate}
 \item  The permutation entropy or the sorting entropy monotonically decrease starting from some order $d_0$.
	In this case by Theorem~\ref{CEofOP_andOrdinalEntr}, the conditional entropy provides a better bound for the KS entropy than the permutation entropy.
  \item The ordinal partition for some order $d$ is generating and has the Markov property, while the system is not order-isomorphic to a Markov shift over two symbols.
	In this case by statement (ii) of Lemma~\ref{MarkovPropertyOfOP_Th} the conditional entropy of order $d$ is equal to the KS entropy.
  \item The ordinal partition has the Markov property for all $d \geq d_0$.
	In this case by statement (i) of Lemma~\ref{MarkovPropertyOfOP_Th} the conditional entropy of order $d$ converges to the KS entropy as $d$ tends to infinity.
\end{enumerate}

\section{Proofs}\label{sectionProof}
In Subsections~\ref{CEofOP_andOrdinalEntr_ProofSect} and \ref{AllPeriodicPointsTh_ProofSect} we give proofs of Theorems \ref{CEofOP_andOrdinalEntr} and \ref{AllPeriodicPointsTh}, respectively.
In Subsection~\ref{sectionMarkov2symbols} we prove Lemma~\ref{theorem2letters}, to do this we use an auxiliary result established in Subsection~\ref{subsectionMarkovGenRemarks}.

\subsection{Proof of Theorem \ref{CEofOP_andOrdinalEntr}}\label{CEofOP_andOrdinalEntr_ProofSect}
\begin{proof}
    For any given partition ${\mathcal P}$, the difference $H({\mathcal P}_{n+1}) - H({\mathcal P}_{n})$ decreases monotonically with increasing $n$ \cite[Section~4.2]{CoverThomas2006}.
    In particular, for the ordinal partition ${\mathcal P}^{\mathbf X}(d)$ it holds
    \begin{align*}
	h_\mu(T, {\mathcal P}^{\mathbf X}(d)) = \lim_{n\to\infty} \mleft( H({\mathcal P}^{\mathbf X}(d)_{n+1}) - H({\mathcal P}^{\mathbf X}(d)_{n}) \mright) 
								&\leq H({\mathcal P}^{\mathbf X}(d)_2)     - H({\mathcal P}^{\mathbf X}(d))\\ &= h_{\mu, \text{cond}}^{\mathbf X}(T, d)
    \end{align*}
    and consequently
    \begin {equation*}
	\lim_{d\to\infty}h_\mu(T, {\mathcal P}^{\mathbf X}(d)) \leq \varlimsup_{d\to\infty}h_{\mu, \text{cond}}^{\mathbf X}(T, d).
    \end {equation*}
    The last inequality together with \eqref{ordinalRepresentationKS} and \eqref{SEandCEofOPineq} implies {\it(i)}.

    The statement {\it(ii)} will be proved once we prove the inequality below:
    \begin {equation}\label{PE_SE_ineq}
	\varlimsup_{d \to \infty}h_{\mu, \triangle}^{\mathbf X}(T, d) \leq \varlimsup_{d \to \infty}h_\mu^{\mathbf X}(T, d).
    \end {equation}
    If \eqref{PEdecreases} is satisfied, we get \eqref{PE_SE_ineq} immediately from Lemma~\ref{SEandCEofOP}.
    If \eqref{SElimitExists} is satisfied, by Cesaro's mean theorem \cite[Theorem~4.2.3]{CoverThomas2006} it follows that
    \begin {equation*}
	\lim_{d \to \infty} \frac{1}{d} H({\mathcal P}^{\mathbf X}(d)) = \lim_{d \to \infty} \mleft (H({\mathcal P}^{\mathbf X}(d+1)) - H({\mathcal P}^{\mathbf X}(d)) \mright),
    \end {equation*}
    which is a particular case of \eqref{PE_SE_ineq}, and we are done.
\end{proof}

\subsection{Proof of Theorem \ref{AllPeriodicPointsTh}}\label{AllPeriodicPointsTh_ProofSect}
\begin {proof}
  Since the KS entropy for periodic dynamical systems is equal to zero, it remains to show that it holds
  \begin {equation}\label{CEofOPisZero}
      h_{\mu, \text{cond}}^{\mathbf X}(T, d) = H({\mathcal P}^{\mathbf X}(d)_2) - H({\mathcal P}^{\mathbf X}(d)) = 0.
  \end {equation}
  We prove this now for a one-dimensional random vector ${\mathbf X} = X$, the proof in the general case is completely resembling.

  The idea of the proof is as follows. 
  Consider periodic points with period not exceeding $k$ that belong to some element of ordinal partition ${\mathcal P}^{X}(d)$ for $d \geq k$. 
  We show below that all these points have the same period and, consequently, that the images of all these points are in the same element of the ordinal partition.
  This means that every ordinal pattern of order $d$ has a completely determined successor, which provides \eqref{CEofOPisZero}.

  By assumption there exists a set $\Omega_0 \subset \mathbb{B}(\Omega)$ such that $\mu(\Omega_0) = 1$ and for all $\omega \in \Omega_0$ 
  it holds $T^{\circ l}(\omega) = \omega$ for some $l \in \{1, 2,\ldots, k\}$.
  Let us fix an order $d \geq k$ and take ordinal patterns $\pi, \xi \in \Pi_d$ such that $\mu(P_\pi \cap T^{-1}(P_\xi)) > 0$.
  We aim to prove that it holds
  \begin {equation*}
      \mu(P_\pi \cap T^{-1}(P_\xi)) = \mu(P_\pi).
  \end {equation*} 
  To do this, let us consider some $\omega_1 \in \Omega_0 \cap P_\pi \cap T^{-1}(P_\xi)$, hence $\pi$ and $\xi$ are the ordinal patterns of the vectors 
  \begin {equation*}
    (X(T^{\circ d}(\omega_1)),X(T^{\circ {d-1}}(\omega_1)),\ldots ,X(\omega_1))
  \end {equation*}
  and
  \begin {equation*}
    (X(T^{\circ (d+1)}(\omega_1)),X(T^{\circ {d}}(\omega_1)),\ldots,X(T(\omega_1))),
  \end {equation*}
  respectively.
  Since $\omega_1 \in \Omega_0$, it is periodic with some (minimal) period $l \in \{1, 2,\ldots, k\}$ such that $X(\omega_1) = X(T^{\circ l}(\omega_1))$.
  Together with Definition~\ref{OrdPatternDef} of an ordinal pattern, this implies $\pi = (\ldots, d, (d-l), \ldots)$.
  Now it is clear that all points $\omega \in \Omega_0 \cap P_\pi$ have the same period. 
  Indeed, the ordinal pattern of any point with period $l_2 \leq k$ such that $l_2 \neq l$ is $(\ldots, d, (d-l_2), \ldots) \neq \pi$.
  Therefore for all $\omega \in \Omega_0 \cap P_\pi$ it holds $X(T^{\circ (d+1)}(\omega)) = X(T^{\circ (d+1 - l)}(\omega))$ and the ordinal pattern for the vector
  \begin {equation*}
    ((X(T^{\circ (d+1)}(\omega)),X(T^{\circ {d}}(\omega)),\ldots,X(T(\omega)))
  \end {equation*}
  is obtained from the ordinal pattern $\pi$ in a well-defined way \cite{KellerSinnEmonds2007}:
  by deleting the entry $d$, adding $1$ to all remaining entries and inserting the entry $0$ to the left of the entry $l$. 
  Since $T(\omega_1) \in P_\xi$, for every other $\omega \in \Omega_0 \cap P_\pi$ it also holds $T(\omega) \in P_\xi$.
  Hence for all $\pi, \xi \in \Pi_d$ with $\mu(P_\pi \cap T^{-1}(P_\xi)) > 0$ it holds
  \begin {equation*}
      \mu(P_\pi \cap T^{-1}(P_\xi)) = \mu(\Omega_0 \cap P_\pi \cap T^{-1}(P_\xi)) =  \mu(\Omega_0 \cap P_\pi)  =  \mu(P_\pi), 
  \end {equation*}
  which yields \eqref{CEofOPisZero} and we are done.
\end {proof}

\subsection{Markov property of a partition for the Markov shift}\label{subsectionMarkovGenRemarks}
Hereafter we call the partition ${\mathcal C} = \{C_0, C_1, \ldots, C_l\}$, where $C_0, C_1, \ldots, C_l$ are cylinders, a {\it cylinder partition}.
By the definition of Markov shifts, the cylinder partition is generating and has the Markov property.

To prove Lemma~\ref{theorem2letters} we need the following result.

\begin{lemma}\label{MarkovityNecessaryConditionTh}
  Let $(A^\mathbb{N}, \mathbb{B}_\Pi(A^\mathbb{N}),m, \sigma)$ be a Markov shift over $A = \{0, 1, \ldots, l\}$.
  Suppose that ${\mathcal P} = \{P_0, P_1, \ldots, P_n\}$ is a partition of $A^\mathbb{N}$ with the following properties:
  \begin{enumerate}
      \item[$(i)$]  any $P_i$ is either a subset of some cylinder $C_a$ ($P_i \subseteq C_a$) or an invariant set with zero measure 
		    ($\sigma^{-1}(P_i) = P_i$ and $m\mleft(P_i\mright) = 0$);
      \item[$(ii)$] for any $P_i, P_j\in {\mathcal P}$ it holds either $P_i \cap \sigma^{-1}(P_j) = C_a \cap \sigma^{-1}(P_j)$ for some $a \in A$ or  
		    $m\mleft(P_i \cap \sigma^{-1}(P_j)\mright) = 0$.
  \end{enumerate} 
  Then the partition ${\mathcal P}$ is generating and has the Markov property.
\end{lemma}
Let us first prove the following lemma.
\begin{lemma}\label{MarkovShiftsProperty}
  Suppose that the assumptions of Lemma~\ref{MarkovityNecessaryConditionTh} hold and that $P_i, P_j$ are sets from ${\mathcal P}$ with $m(P_i \cap \sigma^{-1}(P_j)) \neq 0$, 
  $P_i \subseteq C_a$, and $P_j \subseteq C_b$, where $C_a$ and $C_b$ are cylinders.
  Then for any $P \subseteq P_j$, $P \in \mathbb{B}_\Pi(A^\mathbb{N})$ it holds
   \begin {equation}\label{IntersectionDecomposeEq}
	m(P_i \cap \sigma^{-1}(P)) = \frac{m(C_a \cap \sigma^{-1}(C_b))}{m(C_b)} m(P).
   \end {equation}
\end{lemma}
\begin{proof}
  Given $m(P) = 0$, equality \eqref{IntersectionDecomposeEq} holds automatically, thus assume that $m(P) > 0$. 
  As a consequence of assumption $(ii)$ of Lemma~\ref{MarkovityNecessaryConditionTh}, for any $P \subseteq P_j$ we have
   \begin {align*}
	P_i \cap \sigma^{-1}(P) = P_i \cap \mleft(\sigma^{-1}(P_j) \cap \sigma^{-1}(P)\mright) &= \mleft( C_a \cap \sigma^{-1}(P_j)\mright) \cap \sigma^{-1}(P)\\
											       &= C_a \cap \sigma^{-1}(P).
   \end {align*}
   Since $P \subseteq P_j \subseteq C_b$, for all $s \in P$ the first symbol is fixed.
   One can also decompose the set $P$ into union of sets with two fixed elements:
   \begin {equation*}
	P = \bigcup_{i \in I} \mleft( C_b \cap \sigma^{-1}(B_i) \mright),
   \end {equation*}
   where it holds $B_i \subseteq C_i$ for all $i \in I \subseteq A$.
   Then it follows
   \begin {equation*}
	m(P_i \cap \sigma^{-1}(P)) = m(C_a \cap \sigma^{-1}(P)) = m(C_a \cap \sigma^{-1}(C_b) \cap \sigma^{-2}(\bigcup_{i \in I}B_i)).
   \end {equation*}
   Finally, since $m$ is a Markov measure, we get
   \begin {align*}
	m(C_a \cap \sigma^{-1}(C_b) \cap \sigma^{-2}(\bigcup_{i \in I}B_i)) &= \frac{m(C_a \cap \sigma^{-1}(C_b))}{m(C_b)} m(C_b \cap \sigma^{-1}(\bigcup_{i \in I}B_i))\\
									    &= \frac{m(C_a \cap \sigma^{-1}(C_b))}{m(C_b)} m(P).
   \end {align*}
   This completes the proof. 
\end{proof}

Now we come to the proof of Lemma~\ref{MarkovityNecessaryConditionTh}.
\begin{proof}
  By assumption $(i)$, the partition ${\mathcal P}$ is finer than the generating partition ${\mathcal C}$ except for an invariant set of measure zero, hence ${\mathcal P}$ is generating as well.
  To show that ${\mathcal P}$ has the Markov property let us fix some $n \in \mathbb{N}$ and consider $P_{i_0},P_{i_1}, \ldots, P_{i_n} \in {\mathcal P}$ with
   \begin {equation*}
      m \mleft(P_{i_{0}} \cap \sigma^{-1}(P_{i_{1}}) \cap \ldots \cap \sigma^{-\circ n}(P_{i_{n-1}})\mright) > 0. 
   \end {equation*}
   We need to show that the following equality holds:
	\begin {equation*}
	      \frac{m\mleft(P_{i_{0}} \cap \sigma^{-1}(P_{i_{1}}) \cap \ldots \cap \sigma^{-\circ n}(P_{i_{n}})\mright)}
	           {m\mleft(P_{i_{0}} \cap \sigma^{-1}(P_{i_{1}}) \cap \ldots \cap \sigma^{-\circ(n-1)}(P_{i_{n-1}})\mright)} =
	      \frac{m\mleft(P_{i_{n-1}} \cap \sigma^{-1}(P_{i_{n}})\mright)}{m(P_{i_{n-1}})}.
	\end {equation*}
   According to assumption $(i)$, there exist $a_0, a_1, \ldots, a_n \in A$ with $P_{i_k} \subset C_{a_k}$ for all $k = 0,1, \ldots, n$.
   Therefore by successive application of \eqref{IntersectionDecomposeEq} we have:
    \begin {align*}
	 m &   \mleft(P_{i_{0}} \cap \sigma^{-1}(P_{i_{1}}) \cap \ldots \cap \sigma^{-\circ n}(P_{i_{n}})\mright) = \\ 
	   &= m\mleft(P_{i_{1}} \cap \sigma^{-1}(P_{i_{2}}) \cap \ldots \cap \sigma^{-\circ (n-1)}(P_{i_{n}})\mright)\frac{m(C_{a_0} \cap \sigma^{-1}(C_{a_1}))}{m(C_{a_1})} = \ldots\\
	   &= m(P_{i_{n-1}} \cap \sigma^{-1}(P_{i_{n}})) \prod_{k=0}^{n-2} \frac{m\mleft(C_{a_{k}} \cap \sigma^{-1}(C_{a_{k+1}})\mright)}{m(C_{a_{k+1}})}.
    \end {align*}
   Analogously:
    \begin {align*}
	 m &	\mleft(P_{i_{0}} \cap \sigma^{-1}(P_{i_{1}}) \cap \ldots \cap \sigma^{-\circ (n-1)}(P_{i_{n-1}})\mright) = \\ 
	   &= m \mleft(P_{i_{1}} \cap \sigma^{-1}(P_{i_{2}}) \cap \ldots \cap \sigma^{-\circ (n-2)}(P_{i_{n-1}})\mright)\frac{m(C_{a_0} \cap \sigma^{-1}(C_{a_1}))}{m(C_{a_1})} = \ldots\\
	   &= m(P_{i_{n-2}} \cap \sigma^{-1}(P_{i_{n-1}})) \prod_{k=0}^{n-3} \frac{m\mleft(C_{a_{k}} \cap \sigma^{-1}(C_{a_{k+1}})\mright)}{m(C_{a_{k+1}})}\\
			  &= m(P_{i_{n-1}}) \prod_{k=0}^{n-2} \frac{m\mleft(C_{a_{k}} \cap \sigma^{-1}(C_{a_{k+1}})\mright)}{m(C_{a_{k+1}})},
    \end {align*}
    and we are done. 
\end{proof}

\begin{corollary}\label{MarkovityNecessaryConditionCor}
  Let ${\mathcal P} = \{P_0, P_1, \ldots, P_n\}$ and $\widetilde{{\mathcal P}} = \{P \setminus O \mid P \in {\mathcal P}\} \cup \{ O \}$, where $m(O) = 0$ and $\sigma^{-1}(O) = O$,
  be partitions of $A^\mathbb{N}$.
  If $\widetilde{{\mathcal P}}$ satisfies the assumptions of Lemma~\ref{MarkovityNecessaryConditionTh}, then ${\mathcal P}$ is generating and has the Markov property.
\end{corollary}

\subsection{Proof of Lemma~\ref{theorem2letters}}\label{sectionMarkov2symbols}
Now we show that for ergodic Markov shift over two symbols the ordinal partitions are generating and have the Markov property.
The idea of the proof is to construct for an ordinal partition ${\mathcal P}^X(d)$ a partition $\widetilde{{\mathcal P}}^X(d)$ as in Corollary~\ref{MarkovityNecessaryConditionCor}
and show that $\widetilde{{\mathcal P}}^X(d)$ satisfies the assumptions of Lemma~\ref{MarkovityNecessaryConditionTh}.
Then the partition ${\mathcal P}^X(d)$ is generating and has the Markov property by Corollary~\ref{MarkovityNecessaryConditionCor}.

The proof is divided into a sequence of three lemmas.
First, Lemma~\ref{IncreaseDecreaseForShifts} relates the partition ${\mathcal P}^X(1)$ with the cylinder partition $\mathcal{C}$.
Then we construct the partition $\widetilde{{\mathcal P}}^X(d)$ and show in Lemma~\ref{EqualValuesLemma} that it satisfies assumption $(i)$ of Lemma~\ref{MarkovityNecessaryConditionTh}.
Finally, in Lemma~\ref{MarkovPartition2letters} we prove that $\widetilde{{\mathcal P}}^X(d)$ satisfies assumption $(ii)$ of Lemma~\ref{MarkovityNecessaryConditionTh}.

Given $\overline{0} = (0, 0, \ldots, 0, \ldots), \overline{1} = (1, 1, \ldots, 1, \ldots)$, the following holds.
\begin {lemma}\label{IncreaseDecreaseForShifts}
    Let $P_{(0,1)}, P_{(1,0)} \in {\mathcal P}^X(1)$ be elements of the ordinal partition corresponding to the increasing and decreasing ordinal pattern of order $d = 1$, respectively:
    \begin {equation*}
	P_{(0,1)} = \{ s \in \{0, 1\}^\mathbb{N} \mid X(s)    < X(\sigma s)\},\,\, P_{(1,0)} = \{ s \in \{0, 1\}^\mathbb{N} \mid X(s) \geq X(\sigma s)\},
    \end {equation*}
    where $X$ is lexicographic-like.
    Then it holds
    \begin {equation*}
	  P_{(0,1)} = C_0 \setminus \{\overline{0}\}
    \end {equation*}
and
    \begin {equation*}
	  P_{(1,0)} = C_1 \cup \{\overline{0}\}.
    \end {equation*}
\end {lemma}
\begin {proof}
  We show first that for all $s \in C_0 \setminus \{\overline{0}\}$ it holds $X(s) < X(\sigma s)$.
  Indeed, assume $s = (s_0, s_1, \ldots) \in C_0 \setminus \{\overline{0}\}$.
  Then for the smallest $k \in \mathbb{N}$ with $s_k = 1$ it holds $s_j = (\sigma s)_j$ for $j = 0,\ldots, k-1$ and $s_{k-1} < (\sigma s)_{k-1} = s_k$, that is $s \prec \sigma s$.
  Since $X$ is lexicographic-like, this implies $X(s) < X(\sigma s)$.

  By the same reason for all $s \in C_1 \setminus \{\overline{1}\}$ it holds $X(s) > X(\sigma s)$.
  Finally, as one can easily see, $s \in \{\overline{0}\} \cup \{\overline{1}\}$ implies $X(s) = X(\sigma s)$. 
  According to Definition~\ref{OrdPatternDef} of an ordinal pattern, in this case $s \in P_{(1,0)}$ and we are done.
\end {proof}

In order to apply Corollary~\ref{MarkovityNecessaryConditionCor}, consider the set
   \begin {equation*}
	  O =\bigcup_{n=0}^\infty \sigma^{-\circ n} (\{\overline{0}\}).
   \end {equation*}
By Definition~\ref{MarkovShift_def} of a Markov shift $m(C_0), m(C_1) > 0$, hence no fixed point has full measure.
Together with the assumption of ergodicity of the shift, this implies that the measure of a fixed point is zero, thus $m(O) = 0$.
As is easy to check, $\sigma^{-1}(O) = O$.
Therefore, to prove that the partition ${\mathcal P}^X(d)$ is generating and has the Markov property, it is sufficient to show that the partition 
  \begin {equation}\label{modifiedOrdPartition}
     \widetilde{{\mathcal P}}^X(d) = \{P \setminus O \mid P \in {\mathcal P}^X(d)\} \cup \{ O \}
  \end {equation}
satisfies the assumptions of Lemma~\ref{MarkovityNecessaryConditionTh}.

\begin {lemma}\label{EqualValuesLemma}
 Let $d \in \mathbb{N}$ and $\widetilde{{\mathcal P}}^X(d)$ be the partition defined by \eqref{modifiedOrdPartition} for an ergodic Markov shift over two symbols.
 For every $P \in \widetilde{{\mathcal P}}^X(d) \setminus \{ O \}$ it holds
   \begin {equation*}
    P \subset C_{a_0 a_1 \ldots a_{d-1}},
   \end {equation*}
where $C_{a_0 a_1 \ldots a_{d-1}}$ is a cylinder set.
\end{lemma}
\begin {proof}
Consider the partition consisting of cylinder sets:
\begin {equation*}
    {\mathcal C}_{d} = \{C_{a_0 a_1 \ldots a_{d-1}} \mid a_0, a_1, \ldots, a_{d-1} \in \{0, 1\}\},
\end {equation*}
for $d \in \mathbb{N}$.
According to Lemma~\ref{IncreaseDecreaseForShifts}, ${\mathcal P}^X(1)$ coincides with the partition ${\mathcal C} = \{C_0, C_1\}$ except for the only point $\overline{0}$,
consequently for all $d \in \mathbb{N}$, partition ${\mathcal P}^X(1)_d$ coincides with ${\mathcal C}_d$ except for the points from the set $\sigma^{-\circ (d - 1)}( \{\overline{0}\}) \subset O$.
Since $\widetilde{{\mathcal P}}^X(d) \setminus \{O\}$ is finer than ${\mathcal P}^X(1)_d$, we are done.
\end {proof}

\begin {lemma}\label{MarkovPartition2letters}
 Let $d \in \mathbb{N}$ and $\widetilde{{\mathcal P}}^X(d)$ be the partition defined by \eqref{modifiedOrdPartition} for an ergodic Markov shift over two symbols.
 Given $P_i, P_j\in \widetilde{{\mathcal P}}^X(d)$ with $P_i \subset C_{a_0}$ for $a_0 \in \{0, 1\}$, it holds either
\begin {equation*}
      P_i \cap \sigma^{-1}(P_j) = C_{a_0} \cap \sigma^{-1}(P_j)
\end {equation*}
or
\begin {equation*}
      m(P_i \cap \sigma^{-1}(P_j)) = 0.
\end {equation*}
\end{lemma}
\begin{proof}

Fix some $d\in {\mathbb N}$ and let $P_i, P_j\in \widetilde{{\mathcal P}}^X(d)$.
If $P_i = O$ or $P_j = O$, then it follows immediately that $m(P_i \cap \sigma^{-1}(P_j)) = 0$; thus we put $P_i \neq O$, $P_j \neq O$.
Further, let us define the set $P$ as follows
\begin {equation*}
     P = C_{a_0} \cap \sigma^{-1}(P_j) = \{s = (s_0, s_1, \ldots) \mid s_0 = a_0, (s_1, s_2, \ldots) \in P_j\}.
\end {equation*}
It is sufficient to prove that it holds either $P \subset P_i$ or $P \cap P_i = \emptyset$. 
To do this we show that the ordering of $\mleft(X(s), X(\sigma s), \ldots,  X(\sigma^{\circ d}s)\mright)$ is the same for all $s \in P$.
Since $(s_1, s_2, \ldots) \in P_j$, the ordering of $\mleft(X(\sigma s), X(\sigma^{\circ 2}s),\ldots,  X(\sigma^{\circ d}s)\mright)$ is the same for all $s \in P$.
It remains to show that the relation between $X(s)$ and $X(\sigma^{\circ (k)}s)$ for $k = 1, 2, \ldots, d$ is the same for all $s \in P$.

Note that the order relations between $X(\sigma s)$ and $X(\sigma^{\circ (k+1)}s)$ for $k = 1, 2, \ldots, d$ is given by the fact that $\sigma s = (s_1, s_2, \ldots) \in P_j$.
Next, by Lemma~\ref{EqualValuesLemma} for every $P_j$ there exists a cylinder set $C_{a_1 a_2 \ldots a_{d}}$, such that if $(s_1, s_2, \ldots) \in P_j$ then $s_k = a_k$ for $k = 1, 2, \ldots, d$.
Now it remains to consider two cases:

Assume first that $s_0 = a_0 = 0$ and consider $k = 1, 2, \ldots, d$.
If $s_k = 1$ then $X(s) < X(\sigma^{\circ k}s)$.
Further, if $s_k = 0$, then $X(\sigma s) < X(\sigma^{\circ (k+1)}s)$ implies $X(s) < X(\sigma^{\circ k}s)$, and $X(\sigma s) \geq X(\sigma^{\circ (k+1)}s)$ implies $X(s) \geq X(\sigma^{\circ k}s)$.

Analogously, assume that $s_0 = a_0 = 1$ and consider $k = 1, 2, \ldots, d$.
If $s_k = 0$ then $X(s) \geq X(\sigma^{\circ k}s)$.
If $s_k = 1$, then $X(\sigma s) < X(\sigma^{\circ (k+1)}s)$ implies $X(s) < X(\sigma^{\circ k}s)$, and $X(\sigma s) \geq X(\sigma^{\circ (k+1)}s)$ implies $X(s) \geq X(\sigma^{\circ k}s)$.

Therefore all $s \in P$ are in the same set of the ordinal partition, and consequently it holds either $P \subset P_i$ or $P \cap P_i = \emptyset$.
This finishes the proof.
\end{proof}

\section*{Acknowledgment}
This work was supported by the Graduate School for Computing in Medicine and Life Sciences
funded by Germany's Excellence Initiative [DFG GSC 235/1].

\end{document}